\documentclass{IEEEtran}
\usepackage{cite}
\usepackage{amsmath,amssymb,amsfonts}
\usepackage{algorithmic}
\usepackage{algorithm}
\usepackage{graphicx}
\usepackage{textcomp}
\usepackage{dsfont}
\usepackage{color}
\usepackage{hyperref}
\hypersetup{
    colorlinks=true,
    linkcolor=blue,
    filecolor=blue,      
    urlcolor=blue,
    citecolor=cyan,
}
\usepackage{graphicx} 
\usepackage{float} 
\usepackage{subfig} 

\newtheorem{corollary}{\bf Corollary}[section]
\newtheorem{remark}{\textit{Remark}}

\newtheorem{lemma}{\bf Lemma}[section]
\newtheorem{assumption}{\bf Assumption}[section]
\newtheorem{theorem}{\bf Theorem}[section]
\newtheorem{definition}{\bf Definition}[section]
\definecolor{red}{rgb}{0.9,0,0}

\definecolor{blue}{rgb}{0,0,0.9}

\def\BibTeX{{\rm B\kern-.05em{\sc i\kern-.025em b}\kern-.08em
    T\kern-.1667em\lower.7ex\hbox{E}\kern-.125emX}}
\begin{document}
\title{Equilibrium analysis of three-player General Lotto game with leader-follower framework}
\author{Yang Jiao, Dunbiao Niu, and Yiguang Hong*, \IEEEmembership{Fellow, IEEE}
\thanks{This work was supported by the National Key Research and Development Program of China under Grant 2022YFA1004700,  the National Natural Science Foundation of China under Grant 62573319, and the Fundamental Research Funds for the Central Universities.}
\thanks{ Yang Jiao and Yiguang Hong are with Shanghai Research Institute for Intelligent Autonomous System, Tongji University, Shanghai, 201210, China (e-mail: jy0903@tongji.edu.cn, yghong@iss.ac.cn)  }
\thanks{Dunbiao Niu is with Department of Control Science and Engineering, College of Electronics and Information Engineering, Tongji University, Shanghai, 201804, China (email: dunbiaoniu@tongji.edu.cn)}
}

\maketitle

\begin{abstract}

In this paper, we introduce the General Lotto game with a regulator (R-Lotto), a leader-follower extension of the classical two-player General Lotto game. The model captures regulatory interventions in competitive resource allocation, where a regulator first chooses an intervention parameter to influence the subsequent competition between two resource-constrained followers. The intervention parameter represents favoritism toward one of the followers, and the followers then play a general Lotto subgame with favoritism. We derive the followers' equilibrium payoff and characterize the Nash--Stackelberg equilibrium (NSE) intervention of the regulator. We further develop a multi-battlefield R-Lotto model with a regulator budget constraint. In this setting, the follower subgames on different battlefield is decoupled, while the regulator's intervention decisions are coupled through a common budget. Numerical simulations demonstrate the proposed equilibrium characterizations and provide practical decision-making guidance for the regulator. 
\end{abstract}


\section{Introduction}

Game theory for resource allocation \cite{alghunaim2019proximal,yi2016initialization,diaz2025strategic} has attracted significant attention as a powerful framework for modeling strategic interactions among self-interested agents competing over limited resources. To address the diverse characteristics of such problems, a wide range of game-theoretic models have been proposed. For example, an imperfect-information dynamic Stackelberg game is employed in \cite{wei2016imperfect} to allocate cloud resources and maximize profits through strategic bidding, while a cooperative Nash bargaining game is adopted in \cite{zhang2015resource} to achieve fair subchannel and power allocation in cognitive small-cell networks.

Recently, the two-player General Lotto game \cite{hart2008discrete} has emerged as a prominent framework for capturing bilateral competition over limited resources \cite{paarporn2024incomplete,paarporn2024strategically}. As a canonical variant of the Colonel Blotto game \cite{gross1950continuous,roberson2006colonel}, the General Lotto game differs in that players are allowed to employ randomized resource allocations subject to an expected budget constraint. This relaxation provides both analytical tractability and enhanced modeling flexibility.

Notably, the role of a regulator in competitive resource allocation has received increasing attention in recent years \cite{arnold2011regulation,han2007non}, as such agents explicitly account for overall social welfare \cite{vcivckova2022applied}. By actively intervening in the allocation process, regulators can adjust underlying market mechanisms and steer outcomes toward socially desirable objectives. However, most existing studies on the General Lotto game remain focused on two-player competitions, leaving the strategic role of a potential third player largely unexplored. For instance, \cite{paarporn2024reinforcement} considers a two-player General Lotto game with pre-allocation, where one player allocates resources to battlefields in advance.

Most existing studies on equilibrium analysis, including the General Lotto game, focus on two-player settings, typically characterized by Nash equilibrium (NE) \cite{jiao2026marginal} or Stackelberg equilibrium (SE) \cite{fabiani2021local}. Nevertheless, games with more than two players have also been paid more and more attention to for different equilibria \cite{jiao2025bg,xu2024consistency,bacsar1998dynamic,mukaidani2021infinite,zheng2022stackelberg}. This observation naturally motivates the study of three-player General Lotto games, in which the regulator interacts strategically with the two competing players.

In this paper, we introduce the General Lotto game with a regulator, an extension of the General Lotto game that incorporates regulatory intervention into the strategic environment. In this game, the regulator first chooses an intervention parameter that biases the subsequent competition between two resource-constrained followers. The regulator possesses its own payoff function, and acts as a Stackelberg leader who anticipates the followers’ equilibrium responses. The followers then engage in a general Lotto subgame with favoritism.
\begin{figure*}[t]
    \centering
    \includegraphics[width=1\linewidth]{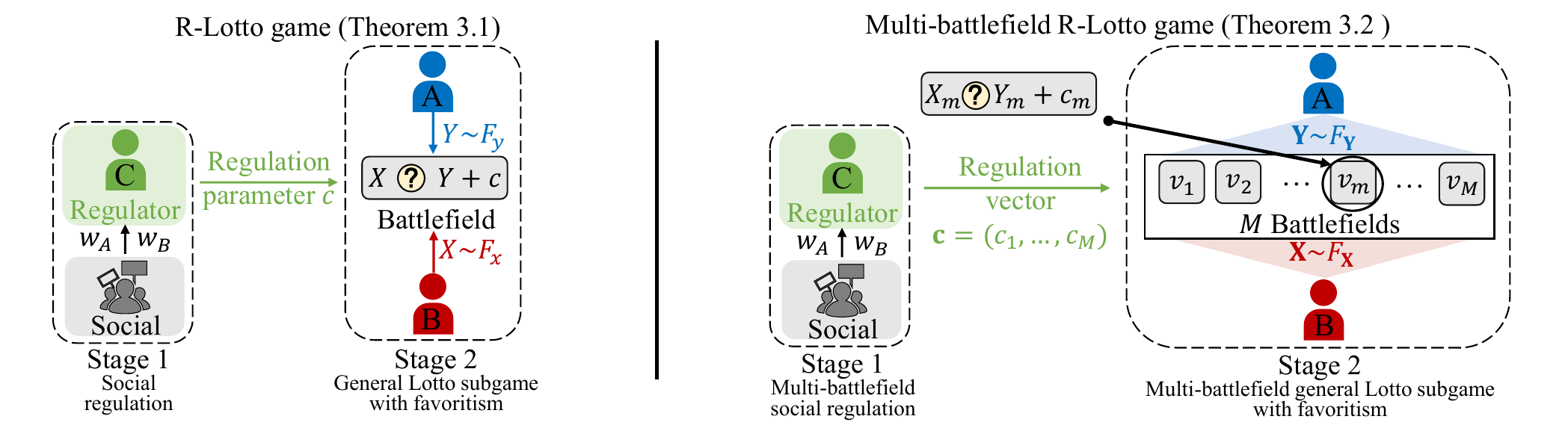}
    \caption{(Left-hand side) The R-Lotto game (See Theorem \ref{the_general_cost_NSE}), where Player C performs a resource pre-allocation, Player A and Player B compete over a single battlefield with favoritism. (Right-hand side) The multi-battlefield R-Lotto game (See Theorem \ref{the_budget_multi_NSE}), extending the single-battlefield to multiple battlefields.}
    \label{fig_R_Lotto}
\end{figure*}

{\color{black}The main contributions are summarized as follows:
\begin{itemize}
    \item For the single-battlefield R-Lotto game, we characterize the NSE under a general intervention cost (Theorem \ref{the_general_cost_NSE}). Specifically, for any continuous and strictly increasing cost function, we reduce the regulator's problem to a bounded one-dimensional optimization problem and derive the corresponding NSE.
    \item We extend the R-Lotto framework to a multi-battlefield setting with an expected regulator budget constraint. By applying Jensen's inequality and the upper concave envelope, we provide the multi-battlefield NSE (Theorem~\ref{the_budget_multi_NSE}). The resulting intervention distribution is either a two-point distribution or a degenerate distribution, and can be efficiently computed by bisection methods.
    \item We provide decision-making guidance for the regulator based on the equilibrium characterizations. In particular, the regulator may focus resources on high-value battlefields while leaving low-value battlefields inactive.
\end{itemize}
}

The rest of the paper is organized as follows: Section II introduces the R-Lotto game. Section III analyzes the corresponding NSE, with an extension to the multi-battlefield case, while Section IV verifies the results through simulations. Finally, Section V concludes the paper.

\section{Game formulation}

In the two-player General Lotto game \cite{hart2008discrete}, two players simultaneously allocate their respective resource budgets across battlefields, and the player who allocates more resources wins.

Here, we extend the Lotto game to a General Lotto game with a regulator (R-Lotto), which is a two-stage leader-follower game that involves two competing players and one social regulator (See Fig. \ref{fig_R_Lotto}). Specifically, the R-Lotto game is organized in the following two stages.

\textit{1) Stage 1 (social regulation):} Player C, acting as a social regulator, chooses an intervention parameter $c \in \mathbb{R}$, which reflects the extent and direction of favoritism in resource pre-allocation. The regulator aims to maximize its own payoff, defined as
\begin{equation}
\label{payoff_c}
    \Pi_C(\Pi_A, \Pi_B; c) := w_A \Pi_A + w_B \Pi_B - H(|c|),
\end{equation}
where two social weight parameters $w_A, w_B \geq 0$ with $w_A + w_B = 1$ indicate Player C's preference between Players A and B, $\Pi_A$ and $\Pi_B$ denote their respective payoffs, and $H(\cdot)$ is the intervention cost function. The objective of Player C is to balance the payoffs of two participants while keeping the cost of intervention as small as possible.  Specifically, a larger $w_A$ (or $w_B$) indicates that Player C pays more attention to A’s (or B’s) payoff. In practice, no intervention incurs no cost, whereas stronger interventions generally require more regulatory resources and therefore lead to higher implementation costs. To capture this phenomenon in a general way, we impose the following mild assumption.
\begin{assumption}
\label{ass_cost}
    The cost function $H(\cdot):\mathbb{R}_{+}\rightarrow \mathbb{R}_{+}$ is continuous and strictly increasing, and satisfies $H(0)=0$, $\lim_{c\rightarrow+\infty}H(c)=+\infty$.
\end{assumption}

\textit{2) Stage 2 (general Lotto subgame with favoritism):} Under the regulator’s intervention, we consider a single-battlefield continuous General Lotto game. Let $a,b \in \mathbb{R}_+$ with $a\ge b >0$ \footnote{Due to the symmetry between Players~A and~B, the alternative case $b\ge a>0$ can be analyzed analogously.} denote the resource budgets of Players A and B, respectively. Player A chooses a cumulative distribution function (CDF) $F_x$ of a nonnegative random variable $X$ with $\mathbb{E}_{X \sim F_x}[X]=a$ \footnote{The resource constraint adopted here follows the formulation in \cite{hart2008discrete}. In recent studies, a relaxed version $\mathbb{E}_{X \sim F_x}[X]\le a$ is also used, which is mathematically equivalent in this context.}, Player B chooses a CDF $F_y$ of a nonnegative random variable $Y$ with $\mathbb{E}_{Y \sim F_y}[Y]=b$. Let $\mathcal{F}_a$ and $\mathcal{F}_b$ denote the strategy sets available to Players~A and~B, respectively. The payoff of Player A is defined as
\begin{equation}
\label{Pi_A}
    \Pi_A(F_x, F_y;c) = \mathbb{P}_{X \sim F_x, Y \sim F_y} ( X \ge c + Y ),
\end{equation}
The payoff of Player B is 
$$\Pi_B(F_x, F_y;c) = 1 - \Pi_A(F_x, F_y;c). $$
Define a general Lotto subgame with favoritism as
\begin{equation}
\label{subgame}
    G=<\{A,B\},\{\mathcal{F}_a,\mathcal{F}_b\},\{\Pi_A,\Pi_B\}>.
\end{equation}

With $\mathcal{C}\subseteq\mathbb{R}$ denoted as the strategy set of Player C, the two-stage R-Lotto game can be written as
\begin{equation}
\label{R_Lotto}
    \Gamma=<\{A,B,C\},\{\mathcal{F}_a,\mathcal{F}_b,\mathcal{C}\},\{\Pi_A,\Pi_B,\Pi_C\}>.
\end{equation}
Then we introduce the NSE for the R-Lotto game as follows.


\begin{definition} \label{Def_NSE}
The strategy triple $(F_x^*,F_y^*; c^*)\in  \mathcal{F}_a \times \mathcal{F}_b\times \mathcal{C}$ is said to be an NSE of the R-Lotto game \eqref{R_Lotto} if $(F_x^*,F_y^*)=(\bar{F}_x^{c^*},\bar{F}_y^{c^*})$ and 
$$
\begin{aligned}
   &  \Pi_C(\bar{F}_x^c,\bar{F}_y^{c};c) \le \Pi_C(F_x^*, F_y^*;c^*), \forall \ c \in \mathcal{C},
\end{aligned}
$$
where $(\bar{F}_x^c,\bar{F}_y^c)$ is the NE of the subgame \eqref{subgame} with intervention parameter $c$, i.e.,
$$
\begin{aligned}
     &\Pi_A(F_x,\bar{F}_y^c;c) \le \Pi_A(\bar{F}_x^c,\bar{F}_y^c;c), \forall \ F_x \in \mathcal{F}_a,\\
     &\Pi_B(\bar{F}_x^c,F_y;c) \le \Pi_B(\bar{F}_x^c,\bar{F}_y^c;c), \forall\ F_y \in \mathcal{F}_b.
\end{aligned}
$$
\end{definition}

\section{Equilibrium Solutions}

{\color{black}In this section, we establish the existence and explicit form of the NSE in the R-Lotto game \eqref{R_Lotto}. We further extend the framework to a multi-battlefield setting. Rather than considering a direct extension of the single-battlefield results, we introduce a regulator budget constraint.

\subsection{The NSE of the R-Lotto game}

According to Definition \ref{Def_NSE}, seeking the NSE of the R-Lotto game can be decomposed into two parts. First, by treating $c$ as a fixed parameter, we need to derive the followers’ NE strategy $(\bar{F}_x^c,\bar{F}_y^c)$ of the subgame \eqref{subgame} and compute their corresponding payoffs, both of which can be regarded as functions of $c$. Second, these payoff functions are substituted into $\Pi_C$ to determine the optimal $c^*$ that maximizes the regulator’s payoff.

To begin with, we introduce the following parameterized families of cumulative distribution functions (CDFs), where $\delta_z$ denotes the Dirac distribution concentrated at $z$ and $U(\ell,r)$ denotes the uniform distribution on $[\ell,r]$.
\begin{itemize}
    \item If $c\ge0$, define 
\begin{equation}
\label{case1}
\begin{aligned}
\hat{F}_x(d) &\sim (1-\tfrac{a}{c+d})\delta_0+\tfrac{a}{c+d}U(c,c+2d), \\
\hat{F}_y(d) &\sim (1-\tfrac{b}{d})\delta_0+\tfrac{b}{d}U(0,2d). 
\end{aligned}
\end{equation}
\item If $c<0$, define\begin{equation}\label{case2}
\begin{aligned}
&\hat{F}_x(d) \sim (1-\tfrac{a}{d})\delta_0+\tfrac{a}{d}U(0,2d), \\
&\hat{F}_y(d) \sim (1-\tfrac{b}{d-c})\delta_0+\tfrac{b}{d-c}U(-c,-c+2d). 
\end{aligned}
\end{equation}
\end{itemize}

The parameter $d$ in \eqref{case1} and \eqref{case2} controls both the continuous support and the probability mass at zero. The next lemma identifies the value of $d$ under which these parameterized distributions form the NE of the follower subgame for each fixed intervention parameter $c$.

\begin{lemma}[Adapted from \cite{vu2021colonel}]
\label{the4}
    Consider the general Lotto subgame with favoritism \eqref{subgame}. For any given intervention parameter $c$, the parameter $d^*$ that characterizes the NE strategy $(\bar{F}_x^c,\bar{F}_y^c)=(\hat{F}_x(d^*),\hat{F}_y(d^*))$ in the parameterized classes \eqref{case1} and \eqref{case2} is given by
        $$
        d^*=
        \left\{
        \begin{aligned}
            & \tfrac{b+\sqrt{b^2+2cb}}{2},&& \text{if}\, c\ge r_0,\\
            &  a-c,&& \text{if}\, 0\le c< r_0,\\
            & \tfrac{a+\sqrt{a^2-2ca}}{2},&& \text{if}\,c<0.
        \end{aligned}
        \right.
        $$
        where $r_0=a-\tfrac{b+\sqrt{b^2+8ab}}{4}$.
\end{lemma}

Based on Lemma \ref{the4}, which provides the closed-form expression of the parameter $d^*$, we can further substitute it into the payoff function defined in \eqref{Pi_A} to obtain the corresponding equilibrium payoff.

\begin{corollary}
\label{cor1}
     Consider the general Lotto subgame with favoritism \eqref{subgame}. For any given intervention parameter $c$, the equilibrium payoff $\Pi_A(\bar{F}_x^c,\bar{F}_y^c;c)$ admits the following closed-form characterization: 
        $$
        \Pi_A=
        \left\{
        \begin{aligned}
            & \tfrac{a}{b+c+\sqrt{b^2+2cb}},&& \text{if}\, c\ge r_0,\\
            &  1-\tfrac{b}{2(a-c)},&& \text{if}\, 0\le c< r_0,\\
            & 1-\tfrac{b}{a-c+\sqrt{a^2-2ca}},&& \text{if}\,c<0.
        \end{aligned}
        \right.
        $$
\end{corollary}

After the exact follower payoff is obtained, the remaining task is to solve the regulator's scalar optimization problem. To simplify notation, throughout the following results we write $P(c):=\Pi_A(\bar F_x^c,\bar F_y^c;c)$. Then the regulator's payoff under the follower equilibrium is
\begin{equation}
\label{leader_objective_general}
    J(c):=\Pi_C(\bar F_x^c,\bar F_y^c;c)=w_B+(w_A-w_B) P(c)-H(|c|).
\end{equation}
The next lemma presents a result of $P(c)$ used by the regulator. Its proof is given in Appendix~\ref{app_payoff_mono}.

\begin{lemma}
\label{lem_payoff_mono}
The function $P(c)$ given in Corollary~\ref{cor1} is continuous and strictly decreasing on $\mathbb{R}$. Moreover, $P(0)=1-\tfrac{b}{2a}$, $\lim_{c\rightarrow-\infty}P(c)=1$, and $\lim_{c\rightarrow+\infty}P(c)=0$.
\end{lemma}

Lemma \ref{lem_payoff_mono} shows that the follower equilibrium payoff $P(c)$ decreases monotonically as the intervention parameter increases. Since the regulator's payoff consists of a weighted follower payoff and an intervention cost, this monotonicity immediately restricts the direction in which an optimal intervention can occur. 

The remaining task is to determine the optimal intervention $c^*$. Using the payoff defined in \eqref{leader_objective_general}, the NSE of the R-Lotto game can be characterized as follows.

\begin{theorem}
\label{the_general_cost_NSE}
Consider the R-Lotto game \eqref{R_Lotto}. Then the NSE intervention $c^*$ is characterized as follows:
\begin{itemize}
    \item If $w_A=w_B$, then $c^*=0$.
    \item If $w_A>w_B$, then
    \begin{equation}
    \label{Cstar_A}
        c^*\in\arg\max_{-c_1\le c \le 0}\{(w_A-w_B) P(c)-H(-c)\},
    \end{equation}
    where $c_1$ is the unique solution of $H(c_1)=(w_A-w_B)\tfrac{b}{2a}$.
    \item If $w_A<w_B$, then
    \begin{equation}
    \label{Cstar_B}
     c^*\in\arg\max_{0 \le c \le c_2}\{(w_A-w_B) P(c)-H(c)\},
    \end{equation}
    where $c_2$ is the unique solution of $H(c_2)=-(w_A-w_B)(1-\tfrac{b}{2a})$.
\end{itemize}
\end{theorem}
\begin{IEEEproof}
    See Appendix~\ref{app_general_cost_NSE}.
\end{IEEEproof}
\begin{remark}
Since $H(\cdot)$ is continuous and strictly increasing, the constants $c_1$ and $c_2$ are well-defined and unique, and can be quickly determined via the bisection method.
\end{remark}

Theorem~\ref{the_general_cost_NSE} provides a general characterization of the NSE intervention under an intervention cost satisfying Assumption~\ref{ass_cost}. Although a closed-form expression of $c^*$ is generally unavailable without specifying the form of $H$, Theorem \ref{the_general_cost_NSE} reduces the regulator's decision to a bounded one-dimensional optimization problem. The NSE intervention can be computed efficiently using standard optimization algorithms, such as golden-section search \cite{nocedal2006numerical} or Brent's method \cite{brent2013algorithms}. 

\subsection{The NSE of the multi-battlefield R-Lotto game}

In the classical multi-battlefield General Lotto game \cite{vu2021colonel}, the resource constraints for Players A and B are usually imposed in a coupled form, namely the expected sum of resources over all battlefields cannot exceed the resource upper bound. However, this may imply that the follower NE does not admit a closed-form as in Lemma~\ref{the4}, making subsequent regulator optimization analytically intractable. Therefore, we adopt decoupled resource constraints, where each battlefield has its own resource upper bound. 

However, if Player C continues to use the cost function $H(|c|)$, then the problem becomes fully decoupled, since the intervention assigned to each battlefield can be computed separately. Therefore, we consider a quite practical case: the regulator has a limited intervention budget, namely an upper bound on the expected sum of interventions over all battlefields, rather than an intervention cost. Such a formulation is common in regulatory and inspection scenarios, where public agents must allocate limited enforcement resources across multiple targets \cite{polinsky2000economic,avenhaus2002inspection}.

In this subsection, we remove the cost term $H(|c|)$ and assume instead that Player C is subject to an expected resource constraint on the sum of interventions across all battlefields. Although the resource limits of Players A and B remain decoupled across battlefields, the regulator's decisions across battlefields are tied together by one common budget. Therefore, the regulator can no longer optimize each battlefield independently. This coupling leads to a nontrivial resource-allocation problem for the regulator.

Suppose that there are $M\ge1$ battlefields. Let $ \mathbf v=(v_1,\ldots,v_M)\in\mathbb R_+^M$, $\sum_{m=1}^M v_m=1$, where $v_m$ denotes the value of battlefield $m$. For simplicity, battlefields with $v_m=0$ are omitted from the analysis. Let $a_m\ge b_m>0$ denote the predetermined resource budgets of Players A and B on battlefield $m$, respectively. The total budgets are $ a'=\sum_{m=1}^M a_m$ and $b'=\sum_{m=1}^M b_m$. The multi-battlefield R-Lotto game with a regulator budget is defined as follows.

\textit{1) Stage 1 (budget-constrained social regulation):}  Let $\rho \in \mathbb{R}_+$ denote the resource budget of Player C. Player C chooses a randomized intervention vector $ \mathbf C=(C_1,\ldots,C_M)\in\mathbb R^M$, where $\mathbb E_{\mathbf{C}\sim F_{\mathbf{c}}}[\sum_{m=1}^M |C_m|]\le \rho$. Player C aims to maximize its own payoff, defined as
\begin{equation}
\label{C_payoff_budget_multi}
   \Pi_C(\Pi_A,\Pi_B,\mathbf{C}):= \mathbb E_{\mathbf{C}\sim F_{\mathbf{c}}}\left[w_A\Pi_A+w_B\Pi_B\right].
\end{equation}

\textit{2) Stage 2 (General Lotto game with favoritism):} Given an intervention vector $\mathbf c=(c_1,\ldots,c_M)$. Player A chooses a distribution $F_{\mathbf{X}}=(F_{x_1},\ldots,F_{x_M})$ of a variable $\mathbf{X}\in \mathbb{R}_{\ge 0}^{M}$ with $\mathbb{E}_{X_{m} \sim F_{x_{m}}}[X_{m}]=a_m$, Player B chooses a distribution $F_{\mathbf{Y}}=(F_{y_1},\ldots,F_{y_M})$ of a variable $\mathbf{Y}\in \mathbb{R}_{\ge 0}^{M}$ with $\mathbb{E}_{Y_{m} \sim F_{y_{m}}}[Y_{m}]=b_m$. The payoff of Player A is defined as $$ \begin{aligned} & \Pi_A(F_{\mathbf{X}}, F_{\mathbf{Y}};\mathbf{c}) = \sum_{m=1}^{M}v_m\mathbb{P}_{\substack{X_{m} \sim F_{x_{m}}\\ Y_{m} \sim F_{y_{m}} }} ( X_{m} \ge Y_{m}+c_m ), \end{aligned} $$ The payoff of Player B is $$\Pi_B(F_{\mathbf{X}}, F_{\mathbf{Y}};\mathbf{c}) = 1 - \Pi_A(F_{\mathbf{X}}, F_{\mathbf{Y}};\mathbf{c}). $$ 

The next result characterizes the follower equilibrium for every realized intervention vector. Its proof follows directly from the separability of the decoupled follower subgame.

\begin{corollary} \label{col_multi_epsilon_NE} Consider the follower's multi-battlefield general Lotto subgame with a given intervention vector $\mathbf{c}$. For every battlefield $m$, the marginal pair $(\bar{F}_{x_m}^{c_m},\bar{F}_{y_m}^{c_m})$ in the NE $(\bar{F}_{\mathbf{X}}^{\mathbf{c}},\bar{F}_{\mathbf{Y}}^{\mathbf{c}})$ is the solution of single-battlefield subgame \eqref{subgame} provided by Lemma \ref{the4} with $a=a_m$, $b=b_m$ and $c=c_m$. \end{corollary}

Then for each battlefield $m$, we define
\[
    P_m(c)
    :=\mathbb{P}_{\substack{X_{m} \sim \bar F_{x_m}^c\\ Y_{m} \sim \bar F_{y_m}^c }} ( X_{m} \ge Y_{m}+c ),
\]
where $(\bar F_{x_m}^c,\bar F_{y_m}^c)$ is the single-battlefield follower NE characterized in Lemma~\ref{the4} with $a=a_m$, $b=b_m$, and intervention parameter $c$. By Corollary~\ref{cor1},
\begin{equation}
\label{Pm_def}
P_m(c)=
\begin{cases}
\tfrac{a_m}{b_m+c+\sqrt{b_m^2+2b_m c}},
& c\ge r_m,\\
1-\tfrac{b_m}{2(a_m-c)},
& 0\le c<r_m,\\
1-\tfrac{b_m}{a_m-c+\sqrt{a_m^2-2a_m c}},
& c<0,
\end{cases}
\end{equation}
where $ r_m
    =
    a_m-\tfrac{b_m+\sqrt{b_m^2+8a_m b_m}}{4}$. For a given deterministic intervention vector $\mathbf c$, the follower equilibrium payoff of Player A is, therefore, $P_M(\mathbf c)
    :=
    \sum_{m=1}^M v_m P_m(c_m).$ Hence, for a randomized intervention vector $\mathbf C$, the regulator's payoff is
\begin{equation}
\label{JM_budget_multi}
    J_M(\mathbf C)
    =
    w_B+(w_A-w_B)
    \sum_{m=1}^M v_m\mathbb E_{\mathbf{C}\sim F_{\mathbf{c}}}[P_m(C_m)].
\end{equation}

For the optimal intervention of the regulator, we define 
\begin{equation}
\label{BA_m_def}
    P_m^-(c) := P_m(-c)-P_m(0), \quad c\ge0,
\end{equation}
\begin{equation}
\label{BB_m_def}
    P_m^+(c) := P_m(0)-P_m(c), \quad c\ge0.
\end{equation}
The function $P_m^-$ is increasing and concave. In contrast, $P_m^+$ is increasing, convex on $[0,r_m)$, and concave on $[r_m,+\infty)$. Then we define an upper concave envelope of $P_m^+$ as
\begin{equation}
\label{Pm_plus_cav}
\widehat P_m^+(c)
=
\begin{cases}
P_m^{+'}(\tau_m) c,
&0\le c<\tau_m,\\
P_m^+(c),
&c\ge \tau_m,
\end{cases}
\end{equation}
where $\tau_m\ge r_m$ is the solution of the tangency equation $ P_m^+(\tau_m)
    =
    \tau_m P_m^{+'}(\tau_m)$. Since
$P_m^+$ is concave on $[r_m,+\infty)$, this solution is unique. For any fixed $\lambda<\frac{v_mb_m}{2a_m^2}$, let $p_{m}(\lambda)\ge1$ be the unique solution of
    \begin{equation}
    \label{qA_equation}
        p_{m}(\lambda)(1+p_{m}(\lambda))^3
        =
        \tfrac{4v_mb_m}{\lambda a_m^2}.
    \end{equation}
For any fixed $\eta\in (0,v_mP_m^{+'}(\tau_m))$, let $q_{m}(\eta)\ge \sqrt{1+\tfrac{2 \tau_m }{b_m}}$ be the unique solution of
    \begin{equation}
    \label{qB_equation}
        q_{m}(\eta)(1+q_{m}(\eta))^3
        =
        \tfrac{4v_ma_m}{\eta b_m^2}.
    \end{equation}
Then we are ready to derive the NSE intervention of the multi-battlefield R-Lotto game.

\begin{theorem}
\label{the_budget_multi_NSE}
Consider the multi-battlefield R-Lotto game. Then the NSE distributions $F_{\mathbf c^*}=(F_{c_1^*},\ldots,F_{c_M^*})$ are characterized as follows:
\begin{itemize}
    \item If $w_A=w_B$, then $F_{c_m^*}\sim\delta_0$ for all $m=1,\ldots,M$.

    \item If $w_A>w_B$, then $F_{c_m^*}\sim\delta_{f_m(\lambda_0)}$, where $\lambda_0>0$ is the unique solution of $\sum_{m=1}^{M}-f_m(\lambda_0)=\rho$. The strictly increasing function $f_m(\lambda)$ is defined as
    \begin{equation*}
    \label{zA_lambda}
   f_m(\lambda)
    =
    \begin{cases}
    0,
    &\displaystyle
    \lambda\ge \tfrac{v_mb_m}{2a_m^2},\\
    \displaystyle
    -\tfrac{a_m}{2}\left(p_{m}^2-1\right),
    &\displaystyle
    \lambda< \tfrac{v_mb_m}{2a_m^2}.
    \end{cases}
    \end{equation*}
   
     \item If $w_A<w_B$, then
     \begin{equation*}
F_{C_m^*}\sim
\begin{cases}
\tfrac{g_m(\eta_0)}{\tau_m}\delta_{\tau_m}\hspace{-1mm}+\hspace{-1mm}(1\hspace{-1mm}-\hspace{-1mm}\tfrac{g_m(\eta_0)}{\tau_m})\delta_0, & 0 \le g_m(\eta_0) < \tau_m,\\
\delta_{g_m(\eta_0)}, & g_m(\eta_0) \ge \tau_m ,
\end{cases}
\end{equation*}
     where $\eta_0> 0$ is the unique solution of $\sum_{m=1}^{M}g_m(\eta_0)=\rho$. The strictly decreasing function $g_m(\eta)$ is defined as
    \begin{equation*}
 g_m(\eta)
=
\begin{cases}
0,
&
\eta>v_mP_m^{+'}(\tau_m),\\
 [0,\tau_m],
&
\eta=v_mP_m^{+'}(\tau_m),\\
\tfrac{b_m}{2}(q_m^2-1),
&
0<\eta<v_mP_m^{+'}(\tau_m).
\end{cases}
    \end{equation*}
\end{itemize}
\end{theorem}
\begin{IEEEproof}
See Appendix \ref{app7}.
\end{IEEEproof}
\begin{remark}
    Although Theorem~\ref{the_budget_multi_NSE} requires solving several nonlinear equations, the NSE intervention can be computed efficiently. For any fixed $\lambda$ and $\eta$, the left side of \eqref{qA_equation} and \eqref{qB_equation} are strictly increasing. Hence, $p_m$ and $q_m$ can be obtained by a simple bisection method. Similarly, $f_m(\lambda)$ and $g_m(\eta)$ are both strictly monotonic functions. Therefore, $\lambda_0$ and $\eta_0$ can also be computed by bisection. For the case $w_A-w_B<0$, if the bisection variable $\eta$ reaches a breakpoint $v_mP_m^{+'}(\tau_m)$, then the corresponding battlefield can receive any allocation in $[0,\tau_m]$. 
Hence, we only need to check whether the remaining budget can be absorbed by these breakpoint battlefields.
\end{remark}

In Theorem \ref{the_budget_multi_NSE}, the parameters $\lambda_0$ and $\eta_0$ can be viewed as horizontal value lines. The value-weighted marginal gains of battlefields lie above the line are activated, and those below the line receive no intervention, while those exactly on the line are used only when the remaining budget needs to be absorbed. When $w_A>w_B$, Player C favors Player A, so the optimal intervention is negative. Since $a_m\ge b_m$, Player A already has a resource advantage, and the regulator only needs to strengthen this advantage in selected battlefields. Hence, the optimal intervention is deterministic and $F_{c_m^*}$ degenerates at the fixed point $f_m(\lambda_0)$. When $w_A<w_B$, Player C favors Player B, so the optimal intervention is positive. In this case, the intervention may become a two-point distribution: on some battlefields Player C either makes a full-strength intervention $\tau_m$ or gives up the battlefield with intervention $0$.

\section{Decision-making guidance for regulator}

 This section illustrates the results in Theorems~\ref{the_general_cost_NSE} and~\ref{the_budget_multi_NSE} through numerical simulations, and translates them into decision-making guidance for the regulator. 

\subsection{Single-battlefield scenario with typical cost functions}

Consider a simple and commonly used class of costs:
\begin{equation}
\label{power_cost}
    H_n(c)=\tfrac{k}{n}|c|^n,\,
\end{equation}
where $k>0$ and $n\ge1$. The parameter $k$ measures the cost intensity, while $n$ captures how fast the marginal cost increases with the intervention magnitude.

\begin{figure}[t]
    \centering
    \includegraphics[width=0.85\linewidth]{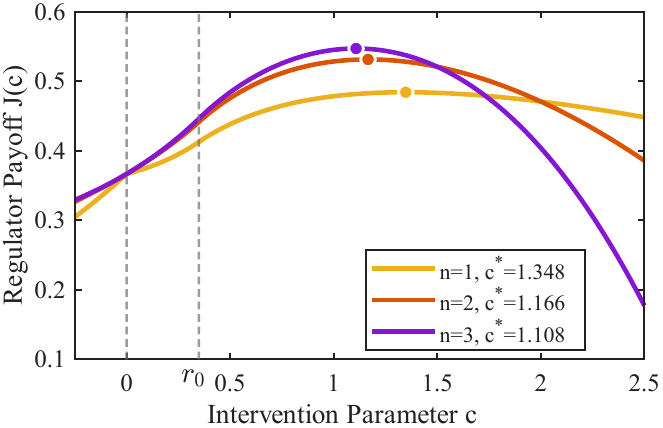}
    \caption{Single-battlefield regulator payoff under $H_n(c)$ defined in \eqref{power_cost}, with $a=1.5,b=1,w_A=0.1,w_B=0.9$.}
    \label{fig1}
\end{figure}

In Fig.~\ref{fig1}, we set $k=0.1$ and compare the regulator payoff at the NSE for $n=1,2,3$, as characterized in Theorem~\ref{the_general_cost_NSE}. All three curves attain their maxima at positive intervention levels, which is consistent with $w_A-w_B<0$. Since the regulator assigns a larger weight to Player B, the optimal decision is to choose $c>0$ and thereby weaken Player A's advantage. The marked optimal points further show how the cost exponent affects the intervention magnitude. As $n$ increases from $1$ to $3$, the optimal intervention decreases from $c_1^*=1.348$ to $c_2^*=1.166$ and then to $c_3^*=1.108$. This indicates that a higher-order cost discourages overly aggressive intervention, because the marginal cost grows faster once $c$ becomes large.

\subsection{Multi-battlefield scenario with resource budget}

Consider $M=5$ battlefields with values
\[
\mathbf v=(0.40,0.25,0.18,0.11,0.06).
\]
Three allocation rules for regulator are compared. The first is the NSE allocation characterized in Theorem~\ref{the_budget_multi_NSE}. The second is the uniform split, where the regulator assigns the same expected intervention budget $\tfrac{\rho}{M}$ to each battlefield. The third is the value-weighted split, where the budget is allocated proportionally to the battlefield values, i.e., the expected intervention on battlefield $m$ is $v_m\rho$.

\begin{figure}[t]
    \centering
    \includegraphics[width=0.85\linewidth]{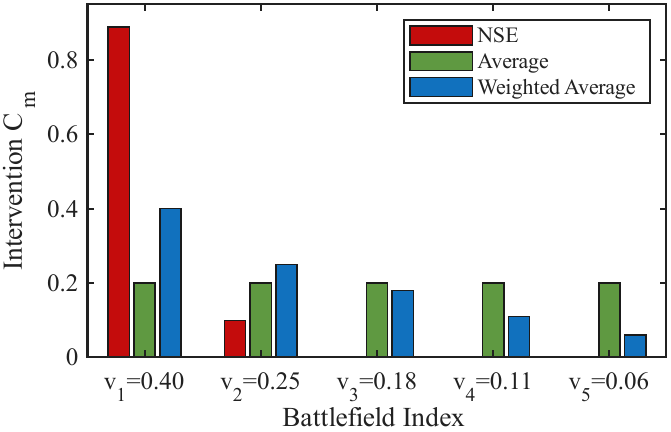}
    \caption{Intervention allocation across five battlefields when the regulator budget $\rho=1$, with $a_m=1.5,b_m=1,w_A=0.1,w_B=0.9$ for all $m$}
    \label{fig_multi_budget_split}
\end{figure}
\begin{figure}[t]
    \centering
    \includegraphics[width=0.85\linewidth]{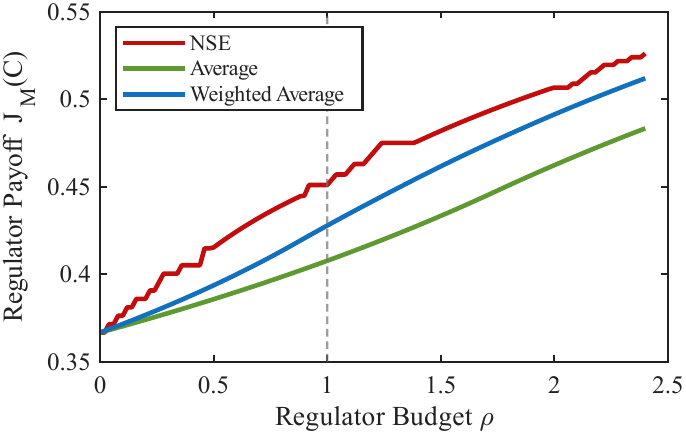}
    \caption{Regulator payoff $J_M(\mathbf C)$ under three allocation rules as the regulator budget $\rho$ varies, with $a_m=1.5,b_m=1,w_A=0.1,w_B=0.9$ for all $m$}
    \label{fig_multi_regulator_payoff}
\end{figure}

Fig.~\ref{fig_multi_budget_split} shows that the NSE allocation is more aggressive than both benchmark rules. While the uniform and value-weighted rules distribute positive intervention to all battlefields, the NSE allocation concentrates most of the budget on the highest-value battlefield and completely abandons the low-value ones. Fig.~\ref{fig_multi_regulator_payoff} compares the resulting payoff of Player C. The NSE achieves the largest regulator payoff over the displayed budget range, and the value-weighted average rule performs better than the simple average rule. 

This pattern is analogous to campaign resource allocation in presidential elections. A campaign planner does not spread intervention resources evenly across all states, nor allocate them strictly in proportion to electoral votes. Instead, the planner concentrates intervention on pivotal swing states and largely leaves states that are safely won or almost impossible to win without additional intervention.}

\section{Conclusion}

In this paper, we introduced the R-Lotto game, an extension of the classical General Lotto game, and derived the NSE. We further extended the R-Lotto framework to a multi-battlefield setting with a regulator budget constraint. The analysis showed that the optimal intervention followed a concentration principle: the regulator allocated resources aggressively to high-value battlefields and abandoned low-value ones. Numerical simulations illustrated the equilibrium structures and translated the theoretical results into decision-making guidance for the regulator. Overall, the R-Lotto framework provided a tractable way to analyze how strategic regulatory intervention affected equilibrium outcomes in both single- and multi-battlefield resource allocation games.

\appendices

\section{The proof of Lemma~\ref{lem_payoff_mono}}
\label{app_payoff_mono}

By Corollary~\ref{cor1}, the equilibrium payoff $P(c)$ has three different expressions according to the value of $c$. We prove the monotonicity of $P(c)$ by checking these three branches separately, and then verify that the pieces are continuously connected at the junction points.

First, consider the case $c<0$. In this region,
\begin{equation}
\label{P_1}
    P(c)=1-\tfrac{b}{a-c+\sqrt{a^2-2ac}}.
\end{equation}
Taking the derivative with respect to $c$ gives
\begin{equation}
\label{P_derivative_negative}
P'(c)=
-\tfrac{b}{(a-c+\sqrt{a^2-2ac})^2}
(1+\tfrac{a}{\sqrt{a^2-2ac}}).
\end{equation}
Since $a>0$, $b>0$, and $\sqrt{a^2-2ac}>0$ for $c<0$, we have $P'(c)<0$. Therefore, $P(c)$ is strictly decreasing on $(-\infty,0)$. Moreover, we can also obtain $\lim_{c\rightarrow 0}P(c)
=
1-\tfrac{b}{2a}$ and $\lim_{c\rightarrow-\infty}P(c)=1$ from \eqref{P_1}.

Second, consider the case $0\le c<r_0$, where
\[
P(c)=1-\tfrac{b}{2(a-c)}.
\]
Thus, the derivative with respect to $c$ is
\begin{equation}
\label{P_derivative_middle}
P'(c)
=
-\tfrac{b}{2(a-c)^2}<0.
\end{equation}
Hence, $P(c)$ is strictly decreasing on $[0,r_0)$. Also, by substituting $c=0$ into this branch, we obtain $P(0)=1-\tfrac{b}{2a}$. Therefore, the left and right limits at $c=0$ coincide.

Third, consider the case $c\ge r_0$. In this case, 
\[
P(c)=\tfrac{a}{b+c+\sqrt{b^2+2bc}}.
\]
Taking derivative with respect to $c$ yields
\begin{equation}
\label{P_derivative_positive}
P'(c)=
-\tfrac{a}{(b+c+\sqrt{b^2+2bc})^2}
(1+\tfrac{b}{\sqrt{b^2+2bc}}).
\end{equation}
Since $a>0$, $b>0$, and $\sqrt{b^2+2bc}>0$ for $c\ge r_0\ge0$, it follows that $P'(c)<0$. Therefore, $P(c)$ is strictly decreasing on $[r_0,+\infty)$. In addition, $\lim_{c\rightarrow+\infty}P(c)=0.$

It remains to verify the continuity at the junction point $c=r_0$. Let $S=\sqrt{b^2+8ab}$. Since $r_0=a-\tfrac{b+S}{4}$,
we have
\[
a-r_0=\tfrac{b+S}{4},
\,
\sqrt{b^2+2br_0}=\tfrac{S-b}{2}.
\]
Substituting these identities into the middle branch and the right branch of $P(c)$ gives
\begin{equation}
\begin{aligned}
1\hspace{-0.5mm}-\hspace{-0.5mm}\tfrac{b}{2(a-r_0)}
\hspace{-0.5mm}=\hspace{-0.5mm}
1\hspace{-0.5mm}-\hspace{-0.5mm}\tfrac{2b}{b+S}\hspace{-0.5mm} =\hspace{-0.5mm}
\tfrac{S-b}{S+b} \hspace{-0.5mm}=\hspace{-0.5mm}
\tfrac{a}{b+r_0+\sqrt{b^2+2br_0}}.
\end{aligned}
\end{equation}
Thus, the two pieces are continuously connected at $c=r_0$.

Combining the above results, $P(c)$ is continuous at both junction points $c=0$ and $c=r_0$, and it is strictly decreasing on each of the three intervals $(-\infty,0)$, $[0,r_0)$, and $[r_0,+\infty)$. Therefore, $P(c)$ is continuous and strictly decreasing on $\mathbb{R}$. This completes the proof.

\section{The proof of Theorem~\ref{the_general_cost_NSE}}
\label{app_general_cost_NSE}
\label{app6}

To prove Theorem~\ref{the_general_cost_NSE}, we first establish a lemma of the regulator's optimal intervention, from the monotonicity of both the follower equilibrium payoff and the intervention cost.

\begin{lemma}
\label{lem_direction_NSE}
For the R-Lotto game \eqref{R_Lotto}, every NSE intervention $c^*$ satisfies:
\begin{itemize}
    \item If $w_A-w_B=0$, then $c^*=0$.
    \item If $w_A-w_B>0$, then $c^*\le0$.
    \item If $w_A-w_B<0$, then $c^*\ge0$.
\end{itemize}
\end{lemma}

\begin{IEEEproof}
Recall that the regulator’s payoff under the follower equilibrium is
\begin{equation*}
    J(c)
    =
    w_B+(w_A-w_B) P(c)-H(|c|),
\end{equation*}
where $P(c)=\Pi_A(\bar F_x^c,\bar F_y^c;c)$.

First, suppose $w_A-w_B>0$. For any $c>0$, Lemma~\ref{lem_payoff_mono} gives $P(c)<P(0)$. Moreover, Assumption~\ref{ass_cost} gives $H(c)>H(0)=0$. Therefore,
\begin{equation}
\begin{aligned}
J(c)
&=w_B+(w_A-w_B) P(c)-H(c)\\
&<w_B+(w_A-w_B) P(0)=J(0).
\end{aligned}
\end{equation}
Thus every positive intervention is strictly dominated by the neutral intervention $c=0$, and no positive intervention can be optimal. Hence, every NSE intervention satisfies $c^*\le0$.

Next, suppose $w_A-w_B<0$. For any $c<0$, Lemma~\ref{lem_payoff_mono} gives $P(c)>P(0)$. Since $H(|c|)>H(0)=0$, we obtain
\begin{equation}
\begin{aligned}
J(c)&=w_B+(w_A-w_B)P(c)-H(|c|)\\
&<w_B+(w_A-w_B) P(0)=J(0).
\end{aligned}
\end{equation}
Thus every negative intervention is strictly dominated by $c=0$, and no negative intervention can be optimal. Hence, every NSE intervention satisfies $c^*\ge0$.

Finally, if $w_A-w_B=0$, then $J(c)=w_B-H(|c|).$ For every $c\ne0$, $H(|c|)>H(0)=0$, and hence $J(c)<J(0)$. Therefore, the unique NSE intervention is $c^*=0$.
\end{IEEEproof}

As a result, the regulator's optimization problem can be decomposed into two separate cases, namely $w_A-w_B>0$ and $w_A-w_B<0$. Lemma~\ref{lem_direction_NSE} substantially reduces the complexity of the equilibrium analysis and serves as the basis for the NSE characterization presented next.

We now prove Theorem~\ref{the_general_cost_NSE}. For each fixed $c$, the follower subgame admits the NE $(\bar F_x^c,\bar F_y^c)$ provided in Lemma~\ref{the4}. Therefore, by Definition~\ref{Def_NSE}, it remains to solve
\begin{equation}
\label{leader_scalar_app}
    \max_{c\in\mathbb{R}} J(c),
\end{equation}
where $J(c)$ is given in \eqref{leader_objective_general}.

If $w_A-w_B=0$, Lemma~\ref{lem_direction_NSE} directly gives $c^*=0$.

Consider the case $w_A-w_B>0$. By Lemma~\ref{lem_direction_NSE}, it is sufficient to search over $c\le0$. For any $c\le0$, we have
\begin{equation}
\begin{aligned}
J(c)-J(0)
&=
(w_A-w_B)\bigl(P(c)-P(0)\bigr)-H(-c).
\end{aligned}
\end{equation}
Moreover, by Lemma~\ref{lem_payoff_mono},
$$
0\le P(c)-P(0)\le 1-P(0)=\tfrac{b}{2a}, \quad c\le0.
$$
Therefore,
$$(w_A-w_B)\bigl(P(c)-P(0)\bigr)-H(-c)
\le
(w_A-w_B)\tfrac{b}{2a}-H(-c).$$
By the definition of $c_1$, we have $H(c_1)=(w_A-w_B)\tfrac{b}{2a}$. Since $H(\cdot)$ is strictly increasing, for every $c<-c_1$, we have $H(-c)>(w_A-w_B)\tfrac{b}{2a},$ which implies $(w_A-w_B)\bigl(P(c)-P(0)\bigr)-H(-c)<0.$ On the other hand, $(w_A-w_B)\bigl(P(0)-P(0)\bigr)-H(0)=0.$ Hence no maximizer can lie outside $[-c_1,0]$. Therefore,
$$
c^*\in
\arg\max_{-c_1\le c\le 0}
\{(w_A-w_B)P(c)-H(-c)\}.
$$

Finally, consider the case $w_A-w_B<0$. By Lemma~\ref{lem_direction_NSE}, it is sufficient to search over $c\ge0$. For any $c\ge0$, we have
\begin{equation}
\begin{aligned}
J(c)-J(0)
&=
(w_A-w_B)\bigl(P(c)-P(0)\bigr)-H(c)\\
&=
-(w_A-w_B)\bigl(P(0)-P(c)\bigr)-H(c).
\end{aligned}
\end{equation}
Moreover, by Lemma~\ref{lem_payoff_mono},
$$
0\le P(0)-P(c)\le P(0)=1-\tfrac{b}{2a}, \, c\ge0.
$$
By the definition of $c_2$, we have $H(c_2)=-(w_A-w_B)(1-\tfrac{b}{2a})$. Since $H(\cdot)$ is strictly increasing,
$$
H(c)>-(w_A-w_B)(1-\tfrac{b}{2a}),\,\forall c>c_2
$$
which implies $-(w_A-w_B)\bigl(P(0)-P(c)\bigr)-H(c)<0.$ Since $-(w_A-w_B)\bigl(P(0)-P(0)\bigr)-H(0)=0,$ no maximizer lies outside $[0,c_2]$. Therefore,
$$
c^*\in
\arg\max_{0\le c\le c_2}
\{(w_A-w_B)P(c)-H(c)\}.
$$

This completes the proof.

\section{The proof of Theorem~\ref{the_budget_multi_NSE}}
\label{app7}

If $w_A-w_B=0$, then $J_M(\mathbf C)=w_B$ for every feasible randomized intervention. In particular, one may take
$F_{c_m^*}\sim\delta_0$ for all $m=1,\ldots,M$.

Let us suppose $w_A-w_B>0$. Since $P_m(c)$ is decreasing in $c$, any
positive $c$ is dominated by replacing it with zero.
Therefore an optimal intervention can be chosen such that $C_m\le 0$. Then $\sum_{m=1}^M\mathbb E[-C_m]\le\rho,$ and maximizing $J_M$ is equivalent to maximizing $ \sum_{m=1}^M v_m\mathbb E[P_m^-(-C_m)].$ Since $P_m^-$ is increasing and concave, Jensen's inequality gives
\[
    \mathbb E[P_m^-(-C_m)]
    \le
    P_m^-(\mathbb E[-C_m]).
\]
Hence, randomization cannot improve the payoff. Let $ z_m=\mathbb E[-C_m].$ The regulator's problem is reduced to
\begin{equation*}
\label{A_alloc}
    \begin{aligned}
        \min_{\mathbf z\in\mathbb R_{\ge 0}^M} \;
        & -\sum_{m=1}^M v_m P_m^-(z_m) \\
        \text{s.t.} \;
        & \sum_{m=1}^M z_m \le \rho ,
    \end{aligned}
\end{equation*}
which is a convex problem, and the Karush-Kuhn-Tucker (KKT) conditions are necessary and sufficient. Therefore, there exists
$\lambda_0>0$ such that
 \begin{equation*} 
    \begin{cases}
    v_m{P_m^-}'(z_m^*)
    =\lambda_0,
    &z_m^*>0,\\
   v_m{P_m^-}'(z_m^*)
    \le \lambda_0,
    &z_m^*=0,
    \end{cases}
    \end{equation*}
By \eqref{Pm_def} and \eqref{BA_m_def},
\[
    {P_m^-}'(z)
    =
    \tfrac{4a_m^2b_m}
    {\sqrt{a_m^2+2a_mz}
    \left(a_m+\sqrt{a_m^2+2a_mz}\right)^3},
\]
and hence, ${P_m^-}'(0)=\tfrac{b_m}{2a_m^2}$. If $\lambda_0\ge \tfrac{v_mb_m}{2a_m^2},$ then $z_m^*=0$. Otherwise, set $p_m=\sqrt{1+\tfrac{2z_m^*}{a_m}}.$ Then $$z_m^*=\tfrac{a_m}{2}(p_m^2-1),$$ and the equation $v_m{P_m^-}'(z_m^*)=\lambda_0$ is equivalent to
\[
    p_m(1+p_m)^3
    =
    \tfrac{4v_mb_m}{\lambda_0a_m^2}.
\]
Since the term on the left-hand side is strictly increasing for $p_m\ge1$, the
solution $p_m$ is unique. Consequently, $F_{c_m^*}\sim\delta_{f_m(\lambda_0)}.$ Moreover, $\sum_{m=1}^M-f_m(\lambda)$ is continuous and strictly decreasing on the relevant active region. Thus, there exists a unique $\lambda_0>0$ satisfying $\sum_{m=1}^M-f_m(\lambda_0)=\rho.$

Finally, suppose $w_A-w_B<0$. Since $P_m(c)$ is decreasing in $c$, any negative realization of $C_m$ is dominated by replacing it with zero. Therefore, an optimal intervention can be chosen such that $C_m\ge0$. Then maximizing $J_M$ is equivalent to
\begin{equation}
\label{app_p1}
    \begin{aligned}
        \min_{\mathbf C\in\mathbb R_{\ge 0}^M} \;
        & - \sum_{m=1}^M v_m\mathbb E[P_m^+(C_m)] \\
        \text{s.t.} \;
        & \sum_{m=1}^M\mathbb E[C_m]\le\rho.
    \end{aligned}
\end{equation}
For each battlefield $m$, by Jensen's inequality and the definition of the upper concave
envelope $\widehat P_m^+$, we have
\[
    \mathbb E[P_m^+(C_m)]
    \le
    \widehat P_m^+(\mathbb E[C_m]).
\]
Let $z_m=\mathbb E[C_m].$ Then the regulator's payoff is upper bounded by $\sum_{m=1}^M v_m\widehat P_m^+(z_m).$ Conversely, this upper bound is attainable. If $0\le z_m<\tau_m$, choose
\[
    F_{c_m}
    \sim
    \tfrac{z_m}{\tau_m}\delta_{\tau_m}
    +
    (1-\tfrac{z_m}{\tau_m})\delta_0.
\]
Then $\mathbb E[C_m]=z_m$ and
\[
    \mathbb E[P_m^+(C_m)]
    =
    \tfrac{z_m}{\tau_m}P_m^+(\tau_m)
    =
    P_m^{+'}(\tau_m)z_m
    =
    \widehat P_m^+(z_m).
\]
If $z_m\ge\tau_m$, choose $F_{c_m}\sim\delta_{z_m}$. Then
\[
    \mathbb E[P_m^+(C_m)]
    =
    P_m^+(z_m)
    =
    \widehat P_m^+(z_m).
\]
Thus, the original randomized problem \eqref{app_p1} is equivalent to
\begin{equation*}
    \begin{aligned}
        \min_{\mathbf z\in\mathbb R_{\ge 0}^M} \;
        & -\sum_{m=1}^M v_m\widehat P_m^+(z_m) \\
        \text{s.t.} \;
        & \sum_{m=1}^M z_m\le\rho .
    \end{aligned}
\end{equation*}
This is again a convex problem. Hence, the KKT conditions are necessary
and sufficient. Therefore, there exists $\eta_0>0$ such that
\begin{equation*}
    \begin{cases}
    v_m \widehat P_m^{+'}(z_m^*)=\eta_0,
    &z_m^*>0,\\
    v_m \widehat P_m^{+'}(0)\le \eta_0,
    &z_m^*=0.
    \end{cases}
\end{equation*}
By the definition of $\widehat P_m^+$,
\[
\widehat P_m^{+'}(z)
=
\begin{cases}
P_m^{+'}(\tau_m),
&0\le z<\tau_m,\\
P_m^{+'}(z),
&z\ge\tau_m.
\end{cases}
\]
Thus, the optimal allocation satisfies $z_m^*\in g_m(\eta_0),$ where
\[
g_m(\eta)
=
\begin{cases}
0,
&
\eta>v_mP_m^{+'}(\tau_m),\\
[0,\tau_m],
&
\eta=v_mP_m^{+'}(\tau_m),\\
\tfrac{b_m}{2}(q_m^2-1),
&
0<\eta<v_mP_m^{+'}(\tau_m).
\end{cases}
\]
In the last case, $z_m^*\ge\tau_m$. Setting $ q_m=\sqrt{1+\tfrac{2z_m^*}{b_m}},$ we have
\[
    z_m^*=\tfrac{b_m}{2}(q_m^2-1).
\]
Moreover, for $z\ge\tau_m$,
\[
    P_m^{+'}(z)
    =
    \tfrac{4a_mb_m^2}
    {\sqrt{b_m^2+2b_mz}
    \left(b_m+\sqrt{b_m^2+2b_mz}\right)^3}.
\]
Hence, the KKT equation $v_mP_m^{+'}(z_m^*)=\eta_0$ is equivalent to
\[
    q_m(1+q_m)^3
    =
    \tfrac{4v_ma_m}{\eta_0 b_m^2}.
\]
The left-hand side is strictly increasing for $q_m\ge1$, so $q_m$ is
unique. Since $z_m^*\ge\tau_m$, we also have $ q_m\ge \sqrt{1+\tfrac{2\tau_m}{b_m}}.$ Finally, the optimal randomized intervention is implemented as follows.
If $0\le z_m^*<\tau_m$, let
\[
    F_{c_m^*}
    \sim
    \tfrac{z_m^*}{\tau_m}\delta_{\tau_m}
    +
    (1-\tfrac{z_m^*}{\tau_m})\delta_0.
\]
If $z_m^*\ge\tau_m$, let $F_{c_m^*}\sim\delta_{z_m^*}$. This is exactly the intervention stated in the
theorem. Combining it with the results in
Corollary~\ref{col_multi_epsilon_NE} gives an NSE. 

This completes the proof.

\bibliographystyle{ieeetr}  
\bibliography{reference}

\end{document}